\documentclass[11pt,bezier]{article}
\usepackage{amsmath,graphicx,amssymb,amsfonts}
\usepackage{color}

\newtheorem{prethm}{{\bf Theorem}}

\newenvironment{thm}{\begin{prethm}{\hspace{-0.5
               em}{\bf.}}}{\end{prethm}}

\newtheorem{pretheorem}{{\bf Theorem}}

\newtheorem{prepro}{Proposition}
\newenvironment{pro}{\begin{prepro}{\hspace{-0.5
               em}{\bf.}}}{\end{prepro}}

\newtheorem{prepr}{{\bf Theorem}}

\newtheorem{predefinition}{Definition}
\newenvironment{definition}{\begin{predefinition}{\hspace{-0.5
               em}{\bf.}}}{\end{predefinition}}

\newtheorem{prelem}{Lemma}
\newenvironment{lem}{\begin{prelem}{\hspace{-0.5
               em}{\bf.}}}{\end{prelem}}

\newtheorem{precor}{Corollary}
\newenvironment{cor}{\begin{precor}{\hspace{-0.5
               em}{\bf.}}}{\end{precor}}

\newtheorem{preexam}{Example}
\newenvironment{exam}{\begin{preexam}{\hspace{-0.5
               em}{\bf.}}}{\end{preexam}}

\newtheorem{preremark}{Remark}
\newenvironment{remark}{\begin{preremark}{\hspace{-0.5
               em}{\bf.}}}{\end{preremark}}

\newtheorem{preexample}{{\bf Example}}

\newtheorem{preproof}{{\bf Proof.}}

\newenvironment{proof}[1]{\begin{preproof}{\rm
               #1}\hfill{$\Box$}}{\end{preproof}}

\newcommand{\rk}{{\rm rank}\,}

\title{\bf\ Generalized Gapped-kmer Filters for Robust Frequency Estimation}
\author{{\normalsize  {  M. Mohammad-Noori${}^{ \textrm{a,c,*}}$},  { N. Ghareghani${}^{ \textrm{b,c}}$}, { M. Ghandi${}^{ \textrm{d,}}$\thanks{Corresponding authors}}\,
}\vspace{2mm} \\{\footnotesize{$^{ \textrm{a}}$\it
School of
Mathematics, Statistics and Computer
Science, College of Science, University of Tehran}}\vspace{-2mm}\\{\footnotesize{\it P.O. Box 14155-6455, Tehran,
Iran}}\\
{\footnotesize{$^{ \textrm{b}}$\it Department of Engineering Science, College of Engineering, University of Tehran,}}\vspace{-2mm}\\
{\footnotesize{\em P.O. Box 11165-4563, Tehran, Iran }}\\
{\footnotesize{$^{ \textrm{c}}$\it School of Mathematics, Institute for Research in
Fundamental Sciences {\rm(IPM),}}}\vspace{-2mm}\\{\footnotesize{\it P.O.Box: 19395-5746, Tehran,
Iran}}\\
{\footnotesize{$^{ \textrm{d}}$\it Broad Institute of MIT and Harvard
7 Cambridge Center, 4034C, }}\vspace{-2mm}\\{\footnotesize{\it Cambridge, MA 02142, United States of America}}\vspace{-2mm}\\
\\{\footnotesize Emails: ghareghani@ipm.ir, ghareghani@ut.ac.ir},\vspace{-2mm}\\
  {\footnotesize  morteza@ipm.ir, mmnoori@ut.ac.ir},\vspace{-2mm}\\
  {\footnotesize  mghandi@gmail.com}
 }

\date{}
\begin{document}
\maketitle


\begin{abstract}
In this paper, we study the generalized gapped k-mer filters and derive a closed form solution for their coefficients. We consider nonnegative integers $\ell$ and $k$, with $k\leq \ell$, and an $\ell$-tuple $B=(b_1,\ldots,b_{\ell})$ of integers $b_i\geq 2$, $i=1,\ldots,\ell$. We introduce and study an incidence matrix $A=A_{\ell,k;B}$. We develop a M\"obius-like function $\nu_B$ which helps us to obtain closed forms for a complete set of mutually orthogonal eigenvectors of $A^{\top} A$ as well as a complete set of mutually orthogonal eigenvectors of $AA^{\top}$ corresponding to nonzero eigenvalues.
 The reduced singular value decomposition of $A$ and combinatorial interpretations for the nullity and rank of $A$, are among the consequences of this approach.
 We then combine the obtained formulas, some results from linear algebra, and combinatorial identities of elementary symmetric functions and $\nu_B$, to provide the entries of the Moore-Penrose pseudo-inverse matrix $A^{+}$ and the Gapped k-mer filter matrix $A^{+} A$.
\end{abstract}

\section{Introduction}

Sequences of length $k$, commonly referred to as $k$-mers, are used in many computational biology algorithms. We previously showed that robust frequency estimation of $k$-mers using gapped $k$-mer features could profoundly improve the performance of algorithms used for sequence classification in computational biology \cite{kmer-b,Enhanced kmer-b}. The method described in these previous publications was based on analytically deriving the coefficients of a gapped $k$-mer filter that could be used to find the robust frequency estimates of $k$-mers. Although this filter could be applied to datasets consisting of DNA or Protein sequences, it was not applicable to complex datasets that included sequences defined on more heterogeneous feature spaces. Here, we provide the closed-form solution for a generalized gapped $k$-mer filter matrix, by relaxing the constraint that all the features are defined on a fixed-size alphabet.

In order to introduce the main object of this introduction, we briefly mention few definitions and notations here; These are presented in more extent and details in the body of the paper. Given two integers $\ell$ and $k$ with $0\leq k \leq \ell$ and a sequence $B=(b_1,\ldots,b_{\ell})$ of integers $b_i\geq 2$, $i=1,\ldots, \ell$, we associate to them two sets of sequences, $U_{\ell; B}$ and $V_{\ell,k; B}$, a match relation between the elements of these two sets and a corresponding $(0,1)$ matrix $A_{\ell,k;B}$ as below. The set $U_{\ell; B}$ consists of all sequences $x_1 \cdots x_{\ell}$ of integers $x_i$ satisfying $0\leq x_i<b_i$ for $i=1,\ldots,\ell$. The set $V_{\ell,k; B}$ consists of all sequences $y_1 \cdots y_{\ell}$, where each $y_i$ is either an integer satisfying $0\leq y_i<b_i$ or an additional gap symbol denoted as $g$; Furthermore, there are exactly $\ell-k$ occurrences of the gap symbol in any $y_1 \cdots y_{\ell} \in V_{\ell,k; B}$. Two sequences $x_1 \cdots x_{\ell} \in U_{\ell;B}$ and $y_1 \cdots y_{\ell} \in V_{\ell,k; B}$ are then matchable if for any $i$, $1 \leq i \leq \ell$, we have $y_i=x_i$ or $y_i=g$. In other words, the gap symbol $g$ acts as a wildcard and can match to any symbol. The corresponding $(0,1)$ matrix $A_{\ell,k ;B}$ is then obtained by indexing its columns and rows respectively by the elements of $U_{\ell;B}$ and $V_{\ell,k;B}$ and setting $A_{\ell,k;B}(v,u)=1$ if and only if $u$ and $v$ are matchable.

When $b_1=\cdots=b_{\ell}=b$, we have a fixed $b$-letter alphabet $\Sigma_b$ and we use the name $A_{\ell,k;b}$ instead of $A_{\ell,k;B}$. In computational biology for DNA sequences, we have $b=4$,
and $\Sigma_4=\{\texttt{A,C,G,T}\}$ is the set of four DNA bases. Then the set of column and row indexes have special names: The set of column indexes, $\Sigma_4^{\ell}$ is non-gapped oligomers of length $\ell$, briefly called non-gapped $\ell$-mers and the set of row indexes is gapped oligomers with $k$ non-gapped positions and length $\ell$, briefly called gapped $k$-mers (of length $\ell$). For amino acid sequences,  $b=20$, $\Sigma_{20}$ is the set of the 20 amino acids, and the column indexes and row indexes are the ungapped and gapped polypeptide sequences of length $\ell$. Apart from some previous studies of $A_{\ell,k;b}$ in mathematics (see \cite{Delsarte,Terwilliger,Delsarte association}), this matrix has recently found profound applications in the field of computational biology and machine learning \cite{kmer-b,Enhanced kmer-b}. Specifically, the inherent symmetry in matrix  $A_{\ell,k;b}$  allowed finding simple closed-form solutions for two related matrices: $W_{\ell,k;b}$ and $H_{\ell,k;b}$, where $W_{\ell,k;b}=A_{\ell,k;b}^{+}$ is the Moore-Penrose pseudo-inverse of $A_{\ell,k;b}$, and $H_{\ell,k;b}$ is the idempotent matrix given by $H_{\ell,k;b}=W_{\ell,k;b}A_{\ell,k;b}$.  In \cite{kmer-b} the matrix $W_{\ell,k;b}$ was derived and used to find robust estimates for $\ell$-mer counts; This led to significant improvement to predict the binding of certain transcription factors to DNA sequences. This work was then extended in \cite{Enhanced kmer-b} and the matrix $H_{\ell,k;b}$ was used to develop a method to efficiently compute the $\ell$-mer count estimates and to compute a string kernel based on these robust count estimates to identify enhancer sequences. Beyond modeling enhancer sequences in mammalian genomes, this method has been widely applied to several problems in computational biology including prediction of the effect of non-coding variants \cite{deltaSVM}, identification of local sequence features influencing cis-regulatory activity \cite{Chaudhari2018}, identification of accessible chromatin regions \cite{Nathans}, and estimation of evolutionary distances for phylogeny reconstruction \cite{Morgenstern}.

In all the above applications, the features were defined over a fixed alphabet length ($b=4$ for DNA/RNA and $b=20$ for amino acids). Here, we show that this constraint could be relaxed to allow generalizing this method to cases with mixture of features that are defined over alphabets of different sizes.  For example, in addition to the DNA sequence that is defined over the alphabet \{\texttt{A,C,G,T}\}, one can also add DNA methylation status which is defined over \{\textit{methylated}, \textit{unmethylated}\} alphabet or other discrete features. Then a similar methodology described in \cite{kmer-b} and \cite{Enhanced kmer-b} can be applied to find a robust estimate of the joint distribution of the features using a limited training data. To achieve this, we take a similar approach as was used in \cite{kmer-b}. We introduce a M\"obius-like function $\nu_B$ and use the related identities to obtain eigenvalues of $A_{\ell,k;B}A_{\ell,k;B}^{\top}$ in terms of elementary symmetric functions.Then we provide a complete set of mutually orthogonal eigenvectors of $A_{\ell,k;B}A_{\ell,k;B}^{\top}$ as well as a complete set of mutually orthogonal eigenvectors of $A_{\ell,k;B}A_{\ell,k;B}^{\top}$ corresponding to the nonzero eigenvalues. This gives the reduced SVD (reduced singular value decomposition) of $A_{\ell,k;B}$. We also give a combinatorial interpretations for the nullity and rank of $A_{\ell,k;B}$ via finding concrete bases for the null space and row space of this matrix. Finally, we derive an equation for the entries of matrices $W_{\ell,k;B}$ and $H_{\ell,k;B}$, where $W_{\ell,k;B}=A_{\ell,k;B}^{+}$ is the Moore-Penrose pseudo-inverse of $A_{\ell,k;B}$ and $H_{\ell,k;B}=W_{\ell,k;B}A_{\ell,k;B}$. Deriving an explicit formula for matrices $A_{\ell,k;B}$ and $H_{\ell,k;B}$ allows efficient computation of robust count estimates from a given training data. In practice, even with modest values of $\ell$ and $k$, these matrices have exponentially large dimensions which makes the application of numeric methods unfeasible.

The rest of the paper is organized as following: Introduction of notation and preliminaries is given in Section \ref{notation}: General notations and definitions for sets, strings, sequences, relations and some symmetric polynomials are presented in Section \ref{notationSet}; Some preliminaries from linear algebra are discussed in Section \ref{PreLinAlg}. The function $\nu_B$ and some of its properties is defined and studied in Section \ref{nu sec}; The main results of this section, that is the identities given in Propositions \ref{mu.mu.sum}, \ref{someNuId} and \ref{NuIdentity}, are used in later sections. Using the definition of function $\nu_B$ and also the elementary symmetric polynomials, we propose an orthonormal basis for the eigenspaces of the matrix $A_{\ell,k;B} A_{\ell,k;B}^{\top}$ in Section \ref{eigen sec}. Concrete bases for the null space and the row space of $A_{\ell,k;B}$ are presented in Section \ref{rowspace}. Finally, in Section \ref{w sec} we compute the entries of $W_{\ell,k;B}$ and $H_{\ell,k;B}$.

\section{Notation and Preliminaries}\label{notation}

\subsection{Notation for sets, strings, sequences and relations}\label{notationSet}
\begin{definition} Let $\ell$ be a positive integer. The set $[\ell]$ is defined as $[\ell]=\{1,\ldots,\ell\}$. For a set $X$ and a nonnegative integer $n$, by ${X \choose n}$, we mean the set of all $n$-element subsets of $X$. Thus $|{X \choose n}|={|X| \choose n}$ and $|{[\ell] \choose n}|={\ell \choose n}$.
\end{definition}

\begin{definition}
 A word $x$ on a finite alphabet $\Sigma$, is a sequence $x=x_1\cdots x_{\ell}$ whose elements $x_i$ belong to the set $\Sigma$.
As in {\rm \cite{kmer-b}} for a given integer $b\geq 2$, the sets $\Sigma_b$, $\Delta_b$, $\Gamma_b$ are defined as follows
$$        \Sigma_b = \{0,1,\cdots,b-1\}, \,\,\,\,
		\Delta_b = \Sigma_b \cup \{g\}, \,\,\,\,
		\Gamma_b = \Delta_b \setminus \{0\},
	$$
where $g$ stands for the gap symbol.
\end{definition}

\begin{definition}
Let $B=(b_1,b_2,\ldots,b_{\ell})$ be an $\ell$-tuple of integers $b_i\geq 2$. Define the sets $\Sigma_B$, $\Delta_B$, $\Gamma_B$, $U_{\ell; B}$ and $V_{\ell,k;B}$ as follows
\begin{align*}
        \Sigma_B &= \Sigma_{b_1}\times \cdots \times \Sigma_{b_{\ell}}, \,\,
        \Delta_B = \Delta_{b_1}\times \cdots \times \Delta_{b_{\ell}},\\
        \Gamma_B &= \Gamma_{b_1}\times \cdots \times \Gamma_{b_{\ell}}, \,\,
        U_{\ell;B} = \Sigma_B,\\
        V_{\ell,k;B} &= \{v\in \Delta_B : |v|_g=\ell-k\}, \,\,
        V'_{\ell,k;B}=\{w\in \Gamma_B: |w|_g=\ell-k\},\\
        V_{\ell, \leq k;B}&=\bigcup_{m=0}^k V_{\ell m}, \,\,
        V'_{\ell, \leq k;B}=\bigcup_{m=0}^k V'_{\ell m}
\end{align*}
\end{definition}
A {\it weak partial order} on a set $S$ is a binary relation $\preceq$ on $S$ which is reflexive, transitive and antisymmetric. A set equipped with a weak partial order is called a {\it partially ordered set} or briefly a {\it poset}. If $a \preceq b$ and $a\neq b$  we write $a\prec b$; Then $\prec$ is nonreflexive, transitive and nonsymmetric; Such a relation is called a {\it strong partial order} on $S$. If $\preceq_1$ (resp. $\preceq_2$ ) is a weak partial order on $S$ (resp. $T$), then $\preceq_1 \times \preceq_2$ is a weak partial order on $S\times T$.
 If for any $a$ and $b$ in $S$, either $a\prec b$ or $b\prec a$, then the partial order is called a {\it total order}, or a {linear order}.
 This notation is used in the following definition.

\begin{definition}
  Let $B=(b_1, b_2, \ldots, b_{\ell})$. We define a partial order on the set $\Delta_B$. For this purpose, firstly for any $1\leq i\leq \ell$, we consider the order $\prec_i$ on the set $\Delta_{b_i}$ given by
$$0\prec_i 1\prec_i \ldots \prec_i b_i-1 \prec_i g $$
and consider the order $\preceq_B\, :=\, (\preceq_1\times \ldots \times \preceq_\ell)$ on $\Delta_B$.
\end{definition}

\begin{remark} As it is clear from the definitions of $U_{\ell;B}$ and  $V_{\ell,k;B}$, when we use these notations we specially emphasize on parameters $\ell$ and $k$.
\end{remark}

 \begin{definition} For any word $v \in \Delta_B$ we set $G_v=\{i: 1\leq i \leq \ell, v_i=g\}$ and $\overline{G}_v = [\ell]\setminus G_v$. If $X=\{x_1,\cdots,x_n\}$ is a subset of $\{1,\cdots,\ell\}$ with $x_1<x_2<\ldots<x_n$ then by $B(X)$ we mean $(b_{x_1},\ldots,b_{x_n})$. Especially, if $v\in \Delta_B$ and  $v'\in \Gamma_B$, then
 $B(G_{v})=(b_i)_{i \in G_{v}}$ and $B(G_{v'})=(b_i)_{i \in G_{v'}}$.
 \end{definition}

 \begin{definition}
    Let $B=(b_1,\ldots,b_{\ell})$. We say elements $u\in \Sigma_B$ and $v\in \Delta_{B}$ match (or $u$ and $v$ are matchable) if for any $1\leq i\leq \ell$ with $v_i\neq g$ we have $u_i=v_i$; We denote this by $v\sim u$. The set of the elements $v\in V_{\ell, k;B}$ which are matchable with $u\in \Sigma_B$, is denoted by $M_{\ell, k;B}(u)$.
    The set of elements $u\in \Sigma_B$ which are matchable with $v$, is denoted by $N_{\ell, k;B}(v)$.
\end{definition}

\begin{definition}\label{defA} The matrix $A_{\ell, k;B}$ is defined as a $(0,1)$ matrix whose rows and columns are indexed respectively by the elements of
$ V_{\ell, k;B}$ and $\Sigma_B$ and $A_{\ell, k;B}(v,u)=1$ if and only if $u$ and $v$ are matchable.
\end{definition}

\begin{remark}
Considering the definition \ref{defA}, if we identify each row index  $v \in V_{\ell, k;B}$ with  $N_{\ell, k;B}(v)$,
then the matrix  $A_{\ell, k;B}$ is seen as an incidence matrix, in which the points and blocks are row indexes and column indexes,
respectively.
\end{remark}

\begin{definition}\label{defALq} The matrix $A_{\ell, \leq k;B}$ is defined as the $(0,1)$ matrix obtained by stacking the matrices $A_{\ell,i;B}$ ($i=0,\ldots, k$), one on top of the other; Thus the rows and columns of $A_{\ell, \leq k;B}$ are indexed respectively by the elements of $V_{\ell,\leq k;B}$ and $U_{\ell,B}$.
\end{definition}

Elementary symmetric polynomials are well-studied objects in the study of polynomials ring $k[x_1,x_2,\ldots,x_n]$ (see Chapter $7$ of \cite{cox}). Below we formally mention their definitions; Then we define another symmetric polynomial which is useful in our work. This is followed by an example demonstrating their applications in our work.

\begin{definition} Let $i$ and $n$ be nonnegative integers and let $X=(x_1,x_2,\ldots,x_n)$ be a finite sequence of  variables. The $i$-th elementary symmetric polynomial, denoted as $S_i(X)$, is defined as $S_i(X):=\sum_{I\in {X \choose i}} \prod_{i\in I} x_i$.
\end{definition}

\noindent {\bf{Notation.}} Let $X=(x_1,\ldots,x_n)$ be a finite sequence of numbers and $\alpha$ and $\beta$ be arbitrary numbers. Then we show the sequence $(\beta x_1+\alpha,\ldots,\beta x_n+\alpha)$ by $\beta X+\alpha$.

\begin{definition} Let $i$ and $n$ be nonnegative integers and let $X=(x_1,x_2,\ldots,x_n)$ be a finite sequence of  variables. The expression $R_i(X)$ is then defined as follows:
\begin{equation} \label{Ri-I}
R_i(X)=\sum_{j=0}^i S_j(X-1)
\end{equation}
\end{definition}

\begin{exam}
Let $0\leq k\leq \ell$ be integers and $B=(b_1,\ldots,b_{\ell})$, $u\in \Sigma_B$ and $v\in V_{\ell k;B}$.
Then we have
\begin{align*}
        &|\Sigma_B| =|\Gamma_B|=\prod_{i=1}^{\ell} b_i, \,\,\,  &|V_{\ell,k;B}|= S_k(B), \,\,\,\          &|V'_{\ell,k;B}|=S_k(B-1), \,\,\,&|V_{\ell,\leq k;B}|= R_k(B+1),\\
        &|V'_{\ell,\leq k;B}|=R_k(B), \,\,\,  &|M_{\ell,k}(u)|={\ell \choose k}, \,\,\,
        &|N_{\ell,k}(v)|=\prod_{i\in G_v} b_i, \,\,\,
\end{align*}

\end{exam}

\subsection{Notation and preliminaries from Linear Algebra}\label{PreLinAlg}
    All matrices we concern in this paper are real matrices. The row space of a $A$ is denoted as ${\rm row}(A)$, the column space of $A$ is denoted as ${\rm col}(A)$, and the dimension of the row space
    of $A$ is denoted as ${\rm rank}(A)$.
    The kernel of $A$, denoted as ${\ker}(A)$ and the nullity of $A$ and denoted as ${\rm null} (A)$.
      The matrix $A$ is called diagonalizable in the field of real numbers if there exists a nonsingular real matrix $P$
     such that $A=P\Lambda_0 P^{-1}$ for some diagonal real matrix $\Lambda_0$.
   If $A$ is diagonalizable, then all eigenvalues of $A$ appear on the main diagonal of  $\Lambda_0$ and the columns of $P$ are the corresponding  eigenvectors.
     The set of column vectors of $P$ is called a {\it complete set of eigenvectors}
     of $A$; The set of column vectors of $P$ which correspond to nonzero eigenvalues is called a {\it complete set of nonzero eigenvectors}
     of $A$. If eigenvectors belonging to distinct eigenvalues of the matrix $A$ are mutually orthogonal, then there exists an
     eigendecomposition $A=P\Lambda_0 P^{-1}$ with $P^{-1}=P^{\top}$, we call such a decomposition an {\it orthogonal eigendecomposition}.
Let $A=P \Lambda_0 P^{\top}$ be a orthogonal eigendecomposition for the matrix $A$ and $P=[Q\,N]$ where ${\rm col}(N)={\ker} (A)$. Then $A=Q \Lambda Q^{\top}$, where the matrix $Q$  is obtained by deleting the columns of $P$ which are in ${\ker} (A)$, and
 $\Lambda$ is obtained by deleting the zero columns and zero rows of $\Lambda_0$; we call this decomposition an {\it orthonormal nonzero eigendecomposition}.

It is known that any symmetric real matrix $A$ is diagonalizable on the field of real numbers  and eigenvectors corresponding to distinct eigenvalues of $A$ are orthogonal. Hence, every symmetric real matrix $A$ has an orthonormal nonzero eigendecomposition of the form $A=Q \Lambda Q^{\top}$ with real matrices $\Lambda$ and $Q$. A real symmetric matrix $A$ of order $n$ is positive definite (resp. positive semi-definite) if ${\bf x}^{\top} A{\bf x}>0$ (resp. ${\bf x}^{\top} A{\bf x}\geq 0$) for all nonzero ${\bf x}\in \mathbb{R}^n$. For any matrix $A$, the matrix $A^{\top}A$ is positive semidefinite, and ${\rm rank}(A) = {\rm rank}(A A^{\top})$. Conversely, any positive semidefinite matrix $M$ can be written as $M = A^{\top}A$; this is the Cholesky decomposition. If $A$ is a real matrix, then both $A^{\top}A$ and $AA^{\top}$ are diagonalizable over the field of real numbers.

A {\it singular value decomposition} (SVD) of a matrix $A\in {\mathbb{R}}^{n\times m}$ is a factorization
$A= U \Sigma V^{\top}$ with $\Sigma = {\rm diag} (\sigma_1, \sigma_2, \ldots, \sigma_p)$, $p= \min\{n, m\}$ and $\sigma_1\geq \sigma_2 \geq \ldots \geq \sigma_p \geq 0,$ such that the set of columns of both matrices $U=[{\bf u_1}, {\bf u_2}, \ldots,  {\bf u_n}]\in {\mathbb{R}}^{n\times n}$ and
$V=[{\bf v_1}, {\bf v_2}, \ldots, {\bf v_m}]\in {\mathbb{R}}^{m\times m}$ are orthonormal. The diagonal entries of $\Sigma$ are called {\it singular values}
of $A$.
If ${\rm rank} (A)=r<p$, then the {\it reduced singular value decomposition} (reduced SVD) of $A$ is a factorization
$A= {\hat{U}} {\hat{\Sigma}} {\hat{V}}^{\top}$ with ${\hat{\Sigma}} = {\rm diag} (\sigma_1, \sigma_2, \ldots, \sigma_r)\in {\mathbb{R}}^{r\times r}$
and $\sigma_1\geq \sigma_2 \geq \ldots \geq \sigma_r >0$, such that the matrices $U=[{\bf u_1}, {\bf u_2}, \ldots,  {\bf u_r}]\in {\mathbb{R}}^{n\times r}$ and
$V=[{\bf v_1}, {\bf v_2}, \ldots, {\bf v_r}]\in {\mathbb{R}}^{m\times r}$ are both orthonormal. The following lemma gives the relation between the SVD of matrix $A$ and eigendecomposition of the matrices $AA^{\top}$ and $A^{\top}A$.
\begin{lem}{\rm (\cite{handbooklin}, Section 5.6, Facts 8,9)} \label{SVD AAT}
Let $A \in {\mathbb{R}}^{n\times m}$, then the following facts holds:
\begin{itemize}
\item[\rm (i)] The nonzero singular values of $A$ are the square roots of nonzero eigenvalues of $A^{\top}A$ or $AA^{\top}$.
\item[\rm (ii)] if ${{U}} {{\Sigma}} {{V}}^{\top}$ is a reduced SVD of $A$, then columns of ${V}$ are eigenvectors of
$A^{\top}A$ and columns of ${U}$ are eigenvectors of $AA^{\top}$.\\
\end{itemize}
\end{lem}

The {\it Moore-Penrose pseudo-inverse} of a matrix $A$, denoted by $A^{+}$, is defined as a matrix that satisfies all the following four conditions:
     \begin{equation*}
        AA^{+}A=A, \,\,\,\,\,  A^{+}AA^{+}=A^{+}, \,\,\,\,\, (AA^{+})^{^{\top}}=AA^{+}, \,\,\,\,\, (A^{+}A)^{^{\top}}=A^{+}A
    \end{equation*}

    The Moore-Penrose pseudo-inverse exists and is unique for any given matrix $A$. We have $ A^{+}= (A^{\top} A)^{+} A^{\top}= A^{\top}(AA^{\top} )^{+}$. For further properties of the Moore-Penrose pseudo-inverse see for instance \cite{handbooklin}.
 The two following Lemmas provide the  Moore-Penrose pseudo-inverse of the matrix $A$
based on some nonzero eigendecomposition of $AA^{\top}$. The proof of the first one is straight forward and left to the readers.

The following lemmas provide the Moore-Penrose pseudo-inverse of a matrix $A$ based on some nonzero eigendecomposition of the matrix $AA^{\top}$.

       \begin{lem}\label{positivesemidef}
       Let $B$ be a positive semi-definite real matrix. Then $B$ admits an orthonormal nonzero eigendecomposition of the form $B=Q \Lambda Q^{\top}$. Where
       $Q^{\top}Q=I$.
       Moreover let $B=AA^{\top}$, then we have
      $ A^{\top}QQ^{\top}=A^{\top}$.
       \end{lem}
\begin{proof}
{Using the previous notation, let $AA^{\top}=P\Lambda_0 P^{\top}$ be an orthonormal decomposition for $AA^{\top}$ and $P=[Q\, N]$ where the columns of $N$ are in ${\ker} (AA^{\top})$. The equation $ Q^{\top}Q=I$ is concluded from the orthonormality of the columns of $Q$. If ${\bf y}$ denotes a column of $N$, by
        $A_{\ell k}A_{\ell k}^{\top}{\bf y}=0$ we obtain ${\bf y}^{\top}A_{\ell k}A_{\ell k}^{\top}{\bf y}=0$, hence $||A_{\ell k}^{\top}{\bf y}||=0$, which yields  $A_{\ell k}^{\top}{\bf y}=0$. Thus
       $A^{\top}N=0.$
         Now, from $PP^{\top}=I$ we obtain $QQ^{\top}+NN^{\top}=I$; Multiplying from left by $A^{\top}$ and using $A^{\top}N=0$, we provide  $ A^{\top}QQ^{\top}=A^{\top}$.
}
\end{proof}

\begin{lem}\label{W_nzed_A}
Let $A_{n \times m}$ be a real matrix and let $AA^{\top}=Q\Lambda Q^{\top}$ be a nonzero orthonormal eigendecomposition of $AA^{\top}$.
Then the Moore-Penrose pseudo-inverse of $A$ is given by $W=A^{\top}Q\Lambda^{-1}Q^{\top}$. Moreover, if the all one column vector ${\bf j}=[1\, 1\, \ldots 1]^{\top}$
is an eigenvector of $A^{\top}A$, then $WA{\bf j}={\bf j}$.
\end{lem}
\begin{proof}
{The proof is easily obtained by using Lemma \ref{positivesemidef}.}
\end{proof}


\begin{lem}\label{W_up}
Let $A_{n \times m}$ be a real matrix and suppose that the columns of $\Upsilon$ are a complete set of eigenvectors corresponding to nonzero eigenvalues of $AA^{\top}$. Let the columns of $\Upsilon$ be $c_1,\ldots,c_n$ corresponding to the nonzero eigenvalues $\lambda_1,\ldots,\lambda_n$.
\begin{itemize}
        \item[\rm (i)]
An orthonormal nonzero eigendecomposition $AA^{\top}=Q\Lambda Q^{\top}$ is obtained by setting
 $Q=\Upsilon E$, where $E={\rm diag} (\frac{1}{\|c_i\|})_{1\leq i \leq n}$.
\item[\rm (ii)] If we denote Moore-Penrose pseudo-inverse of $A$
by $W$, then  $W=A^{\top}\Upsilon D \Upsilon^{\top}$, where $D={\rm diag}(\frac{1}{||c_i||^2 \lambda_i})_{1\leq i \leq n}$. Consequently, $W=A^{\top} C$ where the entries of $C$ are given by $C_{ij}=\sum_{k} \frac{\Upsilon_{ik} \Upsilon_{jk}}{||c_k||^2 \lambda_k}$.
\end{itemize}
\end{lem}
\begin{proof}
{\begin{itemize}
        \item[\rm (i)] In order to obtain normal eigenvectors, it is enough to divide column $c_i$ of $\Upsilon$ by its norm, that is to multiply the matrix $\Upsilon$ from right by the diagonal matrix $E={\rm diag}(\frac{1}{||c_i||})_{1\leq i\leq n}$ to get $Q=\Upsilon E$.
            \item[\rm (ii)] By Lemma \ref{W_nzed_A},
$W=A^{\top}Q \Lambda^{-1} Q^{\top}=A^{\top}\Upsilon E \Lambda^{-1} E^{\top} \Upsilon^{\top}$. Since both $E$ and $\Lambda$ are diagonal, so is $E \Lambda^{-1} E^{\top}$; Setting $D=E\Lambda^{-1} E^{\top}$, we obtain $W=A^{\top} \Upsilon D \Upsilon^{\top}$, where $D={\rm diag}(\frac{1}{||c_i||^2 \lambda_i})_{1\leq i \leq n}$; Setting $C=\Upsilon D \Upsilon^{\top}$ we obtain $W=A^{\top} C$ and $C_{ij}=\sum_{k} \frac{\Upsilon_{ik} \Upsilon_{jk}}{||c_k||^2 \lambda_k}$, as required.
\end{itemize}}
\end{proof}

\section{The function $\nu_B$ and some of its properties}\label{nu sec}

In this section, we consider an order on the set $\Delta_B$ and based on this define a function $\nu_B$ on the set $\Delta_B \times \Delta_B$ and inspect some of its properties. For an integer $b_i\geq 2$, the following linear order makes $\Delta_{b_i}$ a totally ordered set:
$$0\prec_i 1\prec_i \ldots \prec_i b_i-1\prec_i g$$
and when this order is induced on the product set $\Delta_B=\Delta_{b_1}\times \cdots \times \Delta_{b_\ell}$, a poset is obtained; More precisely, for two elements $x=x_1\cdots x_{\ell}$  and $y=y_1\cdots y_{\ell}$ with $x_i,y_i\in \Delta_{b_i}$, $(1\leq i \leq \ell)$, we have $x\preceq_B y$ if and only if $x_i \preceq_i y_i$ holds for $i=1,\ldots,\ell$. Below is presented the definition of
a useful function on $\Delta_B\times \Delta_B$.

\begin{definition}\label{nui-nuB-def}
Consider the $\ell$-tuple $B=(b_1,b_2,\ldots,b_{\ell})$, where $b_i\geq 2$ is integer for $i=1,\ldots,\ell$.
For any $i$,\, \, $(1\leq i\leq \ell)$, we define the function $\nu_{i}$ on $\Delta_{b_i} \times \Delta_{b_i}$ as
	\[
		\nu_i(x,y)=\left\{
			\begin{matrix}-b_i & \text{if }\,\,  x=y=g,\\
-y & \text{if } x=y\neq g,\\
1 & \text{if }\,\, x\prec y,\\				
0 & \text{if } x\succ y.
			\end{matrix}
			\right.
	\]
Now the function $\nu_B$ is defined on the product set $\Delta_B\times \Delta_B$ by the following product rule
\begin{equation}\label{nuB}
\nu_B(x_1\cdots x_{\ell},y_1\cdots y_{\ell})=\prod_{i=1}^{\ell}\nu_i(x_i,y_i)
\end{equation}
\end{definition}

\begin{remark}\label{incidenceAlg}
The function $\nu$ satisfies the property `` $\nu_B(x,y)=0$ unless $x\preceq_B y$"; This means that it is an
element of the incidence algebra of the poset $\Delta_B$ (For the definition and some examples of this concept, see for instance Chapter 8 of  \cite{CameronNotes}). It is observed that $\nu_B$ satisfies
\begin{equation}\label{mobius type}	
		\sum_{x \preceq  z \preceq  y} \nu_B(x,z)=\left\{
			\begin{matrix}\prod_{i=1}^{\ell}(y'_i-2x'_i) & \text{if }\,\,  x\preceq y,\\
				0 & \text{otherwise. }
			\end{matrix}
			\right.
	\end{equation}
where the values $x'_i$, $(1\leq i\leq \ell),$ are defined
	\[
		x'_i=\left\{
			\begin{matrix} b_i & \text{if }\,\,  x_i=g,\\
x_i & \text{otherwise. }
			\end{matrix}
			\right.
	\]
and $y'_i$'s are defined similarly. The equation (\ref{mobius type}) shows similarities between the function $\nu_B$ and the  M\"obius function of the poset $\Delta_B$.

\end{remark}

 Some useful identities about $\nu_B$ are stated in Proposition \ref{mu.mu.sum}, but before stating this proposition we need some definitions and lemmas.

\begin{definition}\label{a0a1a2a3}  Let $\ell$ be a positive integer, $B=(b_1,\ldots,b_{\ell})$ and let $v', v'' \in \Delta_B$.
 Let $m,n$ be integers with $0\leq m,n \leq \ell$ such that $|G_{v'}|= \ell -n$ and $|G_{v''}|= \ell -m$. Define the sets $A_3, A_2, A_1$ and $A_0$ by
 $A_3=\overline{G}_{v'}\cap \overline{G}_{v''}$, $A_2=G_{v''}\setminus G_{v'}$, $A_1=G_{v'}\setminus G_{v''}$
 and $A_0= G_{v'}\cap G_{v''}$.
 \end{definition}

 \begin{lem} \label{A_is} Let $v',v'' \in \Delta_B$ and the sets $A_0$, $A_1$,$A_2$ and $A_3$ be as in Definition \ref{a0a1a2a3}.
    \begin{itemize}
        \item[\rm (i)] The sets $A_3, A_2, A_1$ and $A_0$ are mutually disjoint and
            $A_0\cup A_1 \cup A_2 \cup A_3=[\ell]$.
              Moreover $A_0\neq [\ell]$ unless $v'=v''=g^{\ell}$.
        \item[\rm (ii)] If $A_1=A_2=\emptyset$, then $A_0=G_v=G_{v'}$ and $\overline{G}_{v'}=\overline{G}_{v''}=A_3$; If furthermore $v'\neq v''$, then there exists $i\in A_3$ such that $v'_i\neq v''_i$
    \end{itemize}
 \end{lem}
\begin{proof}{The proof is straightforward. }
\end{proof}

\begin{lem}\label{mu.mu} Let $w, v',v'' \in \Delta_B$.
  \begin{itemize}
        \item[\rm (i)] If $\nu_B (w,v')\nu_B (w,v'')\neq 0$, then $G_w \subseteq A_0$.
        \item[\rm (ii)] If $G_w \subseteq A_0$, then $\nu_B (w,v')\nu_B (w,v'')=p_3 p_2 p_1 p_0$, where
        \begin{align*}
        p_0=&\prod_{i\in G_w}b_i^2 ,& p_1=\prod_{i\in A_1}{\nu}_i(w_i,v''_i),\\
        p_2=&\prod_{i\in A_2}{\nu}_i(w_i,v'_i),& p_3=\prod_{i\in A_3}{\nu}_i(w_i,v'_i){\nu}_i(w_i,v''_i)
        \end{align*}
    \end{itemize}
\end{lem}
\begin{proof}
{ The proof of part (i) is straightforward. The proof of part (ii) is obtained using
\begin{align*}
    \nu_B (w,v')\nu_B (w,v'')=& \prod_{i=1}^{\ell}{\nu}_i(w_i,v'_i){\nu}_i(w_i,v''_i)\\
    =&  \prod_{j=0}^{3}\prod_{i\in A_j}{\nu}_i(w_i,v'_i){\nu}_i(w_i,v''_i),
\end{align*}
and the definition of $\nu_i$.
}
\end{proof}

\begin{pro}\label{mu.mu.sum} Let $v',v''\in \Gamma_B$. Then
\begin{itemize}
    \item [\rm (i)] ${\displaystyle\sum_{w\in V_{\ell k}}\nu_B(w,v')=
    \left\{
		                     \begin{array}{rl}
		                     		(-1)^{\ell-k} S_{\ell-k}(B),\,\, & \hbox{if $v'=g^{\ell}$,} \\
		                            	\\
		                            	0,\,\,  &\hbox{otherwise. }
		                      \end{array}
		                       \right.}$
    \item [\rm (ii)]${\displaystyle\sum_{w\in V_{\ell k}}\nu_B(w,v')\nu_B(w,v'')=
    \left\{
		                     \begin{array}{rl}
		   	{\displaystyle S_{\ell-k}(B(G_{v'}))\prod_{i\in G_{v'}} b_i \prod_{i\in \overline{G}_{v'}} (v'_i+{v'_i}^2)},\,\, & \hbox{if $v'=v''$,} \\
		                            	\\
		                            	0,\,\,  &\hbox{otherwise. }
		                      \end{array}
		                       \right.}$
\end{itemize}
\end{pro}
\begin{proof}
{\begin{itemize}
    \item [\rm (i)] For $w\in V_{\ell,k}$ we have
    $\nu_B(w,g^\ell)=\prod_{i\in G_w}(-b_i)$,
    hence we obtain
    $$\displaystyle\sum_{w\in V_{\ell k}}\nu_B(w,g^{\ell})=(-1)^{\ell-k} S_{\ell-k}(B)$$
which proves part(i) in the case $v'=g^{\ell}$.

Now suppose that $v' \in V_{\ell;\leq k;B}\setminus \{g^{\ell}\}$, hence for some $1\leq j \leq \ell$,\, $v'_j\neq g$. Then from
    \begin{align*}
    {\displaystyle \sum_{w\in V_{\ell k}}\nu_B(w,v')}=& \sum_{w_i=0}^{b_i-1} \prod_{i=1}^{\ell} \nu_i (w_i,v'_i)\\
   =&\prod_{i=1}^{\ell}\sum_{w_i=0}^{b_i-1}\nu_i (w_i,v'_i),
    \end{align*}
    by using $\sum_{w_j=0}^{b_j-1}\nu_j (w_j,v'_j)=0$,
     the right side is simplified to $0$, as required.

\item [\rm (ii)] To prove this part, observe that if $|A_0|<\ell-k$, each summand in the left, is zero and there is nothing to prove. So, let $|A_0|\geq \ell-k$; Setting
$X_{\ell k}(B,G)=\{w\in V_{\ell k}(B): G_w=G\}$ we obtain
\begin{equation}\label{sumnunux}\sum_{w\in V_{\ell k}(B)}\nu_B(w,v') \nu_B(w,v'')=\sum_{G\in {{A_0}\choose{\ell-k}}} \sum_{w\in X_{\ell k}(B,G)}\nu_B(w,v') \nu_B(w,v'')\end{equation}

 First we compute the summand $\sum_{w\in X_{\ell k}(B,G)}\nu(w,v') \nu(w,v'')$, for a fixed $G\in {{A_0}\choose{\ell-k}}$.
 For this, without loss of generality, let $A_3=\{1, \ldots, a_3\}$, $A_2=\{a_3+1, \ldots, a_3+a_2\}$, $A_1=\{a_3+a_2+1, \ldots, a_3+a_2+a_1\}$ and
 $A_0=\{a_3+a_2+a_1+1, \ldots, \ell\}$, where $a_1,a_2,a_3$ are non-negative integers. Moreover, without loss of generality, let $G=\{k+1, \ldots, \ell\}$. Now $w\in X_{\ell k}(B, G)$ can be factorized in the form $w=qrstg^{\ell-k}$, with $|q|=a_3$, $|r|=a_2$, $|s|=a_1$ and $|t|=a_0 - (\ell-k)$ and when $w$ runs over $X_{\ell k}(B,G)$, each of the words $q,r,s$ and $t$ runs over a proper set accordingly. By part (ii) of Lemma \ref{mu.mu} we obtain
 \begin{equation}
 \label{pis}
  \sum_{w\in X_{\ell k}(B,G)}\nu_B(w,v') \nu_B(w,v'')=P_3 P_2 P_1 P_0,
 \end{equation}
 where
 {\small
 \begin{align*}
        P_0=&\prod_{i\in A_0}b_i \prod_{i\in G}b_i,& P_1=\prod_{i \in A_1}\sum_{w_i=0}^{b_i-1}{\nu}_i(w_i,v'_i),\\
        P_2=&\prod_{i \in A_2}\sum_{w_i=0}^{b_i-1}{\nu}_i(w_i,v''_i),&
        P_3=\prod_{i \in A_3}\sum_{w_i=0}^{b_i-1}{\nu}_i(w_i,v'_i){\nu}_i(w_i,v''_i)
        \end{align*}}
Now, we consider two cases:
\\
{\bf Case 1.} $v'=v''$;
In this case $A_0=G_{v'}$ and $A_1=A_2=\emptyset$, hence $P_1=P_2=1$ and
$\sum_{w\in X_{\ell k}(B,G)}\nu_B^2(w,v')=P_3P_0$.
In this case, $P_3={\displaystyle \prod_{i\in {\overline G}_{v'}} (v'_i+{v'_i}^2)}$ and
$P_0={\displaystyle \prod_{i\in G_{v'}}b_i \prod_{i\in G}b_i}$, so
\begin{align*}\sum_{w\in V_{\ell k}(B)}\nu_B(w,v') \nu_B(w,v'')=&\sum_{G\in {{G_{v'}}\choose{\ell-k}}}  \sum_{w\in X_{\ell k}(B,G)}\nu_B(w,v') \nu_B(w,v'')\\
=&\sum_{G\in {{G_{v'}}\choose{\ell-k}}} \prod_{i\in {\overline G}_{v'}} (v'_i+{v'_i}^2)\prod_{i\in G_{v'}}b_i \prod_{i\in G}b_i\\
=&\prod_{i\in {\overline G}_{v'}} (v'_i+{v'_i}^2)\prod_{i\in G_{v'}}b_i\sum_{G\in {{A_0}\choose{\ell-k}}}  \prod_{i\in G}b_i\\
=& \prod_{i\in {\overline G}_{v'}} (v'_i+{v'_i}^2)\prod_{i\in G_{v'}}b_i\,\,\,\,\, S_{\ell-k}(B(G_{v'}))
\end{align*}\\
{\bf Case 2.}  $v'\neq v''$; If $A_1\neq \emptyset$ then $P_1=0$ and if $A_2\neq \emptyset$ then $P_2=0$. Otherwise, if $A_1=A_2=\emptyset$, then
by Lemma \ref{A_is} (ii), $A_3\neq \emptyset$ and there exists $i\in A_3$ with $v'_i\neq v''_i$; For this $i$,
$\sum_{w_i}\nu_i(w_i,v'_i) \nu_i(w_i,v''_i)=0$ thus $P_3=0$. Hence, the hypothesis $v'\neq v''$ implies that the right side of (\ref{pis}) is zero in either case, and we get the result by (\ref{sumnunux}).
\end{itemize}}
\end{proof}

\begin{pro}\label{someNuId} Let $v'\in \Gamma_B$. Then
\begin{itemize}
\item[\rm(i)] For any $u\in U_{\ell}$ we have
\begin{equation}\label{nuId1}
				\sum_{y\in M_{\ell,k;B}(u)} \nu_B(y,v')=(-1)^{\ell-k}S_{\ell-k}(B(G_{v'}))\, \nu_B(u,v'),
			\end{equation}
\item[\rm(ii)] For any $v\in V_{\ell k}$ we have
\begin{equation}
\label{Az2}
\sum_{u\in N_{\ell, k;B}(v)}\nu_B(u,v')=(-1)^{\ell-k} \nu_B(v,v') \end{equation}
\end{itemize}
\end{pro}
\begin{proof}
    {     \begin{itemize}
    \item[\rm(i)] If the summand $\nu(y,v')$ is nonzero, then $G_y \subseteq G_{v'}$, on the other hand the non-gaped positions of all such $y$'s are the same as $u$.
            Let $G_y= \{x_1, x_2, \ldots, x_{\ell-k}\}$, then
            $\nu_B(y,v')=(-1)^{\ell-k} b_{x_1}b_{x_2}\cdots b_{x_{\ell-k}}\, \nu_B(u,v')$. Therefore
            \begin{align*}\sum_{y\in M_{\ell, k;B}(u)} \nu_B(y,v')=& (-1)^{\ell-k}\nu_B(u,v')\sum_{\{x_1, x_2, \ldots, x_{\ell-k}\}\subseteq G_{v'}}b_{x_1}b_{x_2}\ldots b_{x_{\ell-k}}\\
            =& (-1)^{\ell-k} S_{\ell-k}(B(G_{v'}))\, \nu_B(u,v'),
            \end{align*} as required.
    \item[\rm(ii)]
We distinguish two cases:

 			{\bf Case (a).} Suppose that $G_{v}\subseteq G_{v'}$. Then for any $u\in N_{\ell, k;B}(v)$ we have	
				$\nu_B(u,v')=\prod_{i\in \overline{G}_{v'}}\nu_i(u_i,v'_i)=\prod_{i\in \overline{G}_{v'}}\nu_i(v_i,v'_i)$ and since there are totally  $\prod_{i\in G_{v'}} b_i $ such words $u$, the left side of equation (\ref{Az2}) equals
\begin{align*}
\prod_{i\in G_v} b_i \prod_{i\in \overline{G}_{v'}}\nu_i(v_i,v'_i)&=(-1)^{|G_v|}\prod_{i\in G_{v}}\nu_i(g,g)
\prod_{i \in G_{v'}\setminus G_v} \nu_i(v_i,g) \prod_{i\in \overline{G}_{v'}}\nu_i(v_i,v'_i)\\
&=(-1)^{\ell-k} \prod_{i=1}^{\ell}\nu_i(v_i,v'_i)
\end{align*}
which equals $(-1)^{\ell-k} \nu_B(v,v')$ as required.
			
			{\bf Case (b).} Suppose that $G_{v}\not \subseteq G_{v'}$, consequently $G_v \setminus G_{v'}\neq \emptyset$.
			Now for any $i\in G_v \setminus G_{v'}$ we have $\nu(v_i,v'_i)=0$, thus the right side of (\ref{Az2}) is $0$; The following argument shows that the left side is $0$ as well: The nonzero summands in the left side of (\ref{Az2}) are obtained from elements $u\in X$ where the subset $X\subseteq \Sigma_B$ is given by
			$$X=\{u\in \Sigma_B : u_i\leq v_i {\rm \,\, for \,\,} i\in G_{v}\setminus G_{v'} {\rm \,\,and\,\,} u_i=v_i {\rm \,\, for \,\,} i\in \overline{G}_v \}.$$
			Thus we obtain
			\begin{align*}
				\sum_{u\in N_{\ell k;B}(v)} \nu_B(u,v')&=\sum_{u\in X} \nu_B(u,v')\\
				&=\sum_{u\in X} \prod_{i=1}^\ell  \nu_i(u_i,v'_i)\\
				&=\sum_{u\in X} \left( \prod_{i\in \overline{G}_v} \nu_i(u_i,v'_i) \prod_{i\in G_v \setminus G_{v'}} \nu_i(u_i,v'_i) \prod_{i\in G_v \cap G_{v'}} \nu_i(u_i,g)\right)\\
				&=\left( \prod_{i\in \overline{G}_v} \nu_i(v_i,v'_i) \right) \left( \prod_{i\in G_v \setminus G_{v'}} \sum_{u_i=0}^{v'_i}\nu_i(u_i,v'_i)\right)
				\left( \prod_{i\in G_v \cap G_{v'}} b_i\right)
			\end{align*}
			which is $0$ because for any $i\in G_v\setminus G_{v'}$ we have $\sum_{u_i=0}^{v'_i}\nu_i(u_i,v'_i)=\sum_{u_i=0}^{v'_i-1}1-v'_i=0.$
Thus (\ref{Az2}) is true in either case.
    \end{itemize}
    }
    \end{proof}

\begin{definition} \label{defPQ} Let $u\in \Sigma_B$ and $v\in \Delta_B$. Then $P(u,v)$ and $Q(u,v)$ are defined as below
\begin{align*}
P(u,v)&=\{i:1\leq i\leq \ell , v_i\neq g, v_i=u_i\}\\
Q(u,v)&=\{i:1\leq i\leq \ell , v_i\neq g, v_i\neq u_i\}
\end{align*}
We denote $P(u,v)$ and $Q(u,v)$ by $P$ and $Q$, respectively, if there is no danger of confusion.
\end{definition}
With the above definition, it is obvious that $|P(u,v)|+|Q(u,v)|=|\overline{G}_v|.$
Particularly, if $v\in V_{\ell,k}$ and $|P|=p$ then $|Q|=k-p$.

\begin{pro} \label{NuIdentity} Let $u\in \Sigma_B $ and $v\in \Delta_B$. Recall the notation of Definition \ref{defPQ}.
        \begin{itemize}
            \item[\rm(i)]  For $1\leq i\leq \ell $, let
                $$\phi_i(u,v)=\sum_{j=\max\{1,v_i\}}^{b_i-1} \frac{\nu_i(v_i,j)\nu_i(u_i,j)}{j(j+1)}.$$
            Then we have
                $$\phi_i(u,v)=\left\{
                             \begin{array}{ll}
                                    \frac{b_i-1}{b_i}, & \hbox{if $i\in P$;} \\
                                        \\
                                        \frac{-1}{b_i}, & \hbox{otherwise, i.e. if } i\in Q.
                              \end{array}
                               \right.$$
            \item[\rm(ii)] Let $G$ be a given subset of $[\ell]$ with $G_v \subseteq G$. Then the following identity holds
            \begin{equation}
                \label{fracNuId} {\displaystyle \sum_{v'\in \Gamma_B\,\, ,G_{v'}=G} \frac{\nu_B(v,v')\nu_B(u,v')}{ {\displaystyle \prod_{i\in \overline{G}_{v'}} (v'_i+{v'_i}^2)}}
                }= \frac{\displaystyle (-1)^{|Q\setminus G|+|G_v|} \prod_{i \in G_v} b_i \prod_{i\in P \setminus G}(b_i-1)}{\displaystyle \prod_{i \in \overline{G}}b_i}.
            \end{equation}
        \end{itemize}
    \end{pro}

\begin{proof}{ \begin{itemize}
            \item[\rm(i)]  The proof if this part is easy and left to the reader.

            \item[\rm(ii)]
            Note that if $G_{v} \subseteq G_{v'}=G$, it is easily obtained that
            \begin{equation*} \nu_B(v,v')\,\nu_B(u,v')=(-1)^{|G_v|}\prod_{i\in G_v} b_i \prod_{i\in \overline{G}} \nu_i(v_i,v'_i)\,\nu_i(u_i,v'_i)
            \end{equation*}
            hence if we denote by $S$ the left side of (\ref{fracNuId}), we obtain
            \begin{align*}
               S &=(-1)^{|G_v|} {\displaystyle \prod_{i\in G_v} b_i} {\displaystyle \sum_{v'\in \Gamma_B\,\, ,G_{v'}=G} \,\,\prod_{i\in \overline{G}} \frac{ {\nu_i(v_i,v'_i)\nu_i(u_i,v'_i)}} { {v'_i(v'_i+1)}}
                }\\
                &=(-1)^{|G_v|} {\displaystyle \prod_{i\in G_v} b_i  \prod_{i \in \overline{G}} \,\,\sum_{j=\max\{1,v_i\}}^{b_i-1} \frac{\nu_i(v_i,j)\nu_i(u_i,j)}{j(j+1)}}\\
                &=(-1)^{|G_v|} {\displaystyle \prod_{i\in G_v} b_i  \prod_{i \in \overline{G}}\phi_i(u,v)}
                \end{align*}
                Thus by part (i), we get
                \begin{align*}
                S &=(-1)^{|G_v|} {\displaystyle \prod_{i\in G_v} b_i  \prod_{i \in \overline{G}\cap Q}\frac{-1}{b_i} \prod_{i \in \overline{G}\cap P}\frac{b_i-1}{b_i}}\\
                &=\frac{\displaystyle (-1)^{|Q \setminus G|+|G_v|} \prod_{i \in G_v} b_i \prod_{i\in P\setminus G}(b_i-1)}{\displaystyle \prod_{i \in \overline{G}}b_i}.
                \end{align*}
 \end{itemize}
}
\end{proof}

\section{Orthonormal nonzero eigendecomposition of the matrices $A_{\ell k;B}A_{\ell k;B}^{\top}$ and $A_{\ell k;B}^{\top}A$}\label{eigen sec}

In this section we give an orthonormal nonzero eigendecomposition of the matrices $A_{\ell k;B}A_{\ell k;B}^{\top}$ and $A_{\ell k;B}^{\top}A$. Eigenvectors of these matrices are the elementary symmetric polynomials and the entries of the corresponding eigenvectors are given in terms of the function $\nu_B$. Using the properties of $\nu_B$ we show that these
eigenvectors are mutually orthogonal.

    \begin{definition} Let $n,  k\leq \ell$ be integers. Given $B=(b_1,\ldots,b_{\ell})$ and $v'\in V'_{\ell,n;B}$, we define the column vector ${\bf x}_{v'}^{\ell,k,n}$ as a vector whose rows are indexed by the elements of $V_{\ell,k;B}$ with entries ${\bf x}_{v'}^{\ell,k,n}(w)=(-1)^{\ell-k}\nu_B(w,v')$. The column vector ${\bf z}_{v'}^{\ell,n}$ is then defined as ${\bf z}_{v'}^{\ell,n}={\bf x}_{v'}^{\ell,\ell,n}$; In other words, ${\bf z}_{v'}^{\ell,n}$ is a column vector whose rows are indexed by elements $u$ of $\Sigma_B$ with entries ${\bf z}_{v'}^{\ell,n}(u)=\nu_B(u,v')$. When there is no need to emphasize on the parameters $\ell$, $k$ and $n$, we simply write ${\bf x}_{v'}$ and ${\bf z}_{v'}$.
    \end{definition}

    \begin{pro}\label{p-nrm-x}
Let $v' \in \Gamma_B$ and $n= |{\overline{G}}_{v'}|$. Then the following identity holds:
$$\parallel {\bf x}_{v'}^{\ell,k,n}\parallel^2=  S_{\ell -k}(B(G_{v'}))\prod_{i\in {\overline{G}}_{v'}} (v'_i+{v'_i}^2)\prod_{i\in G_{v'}} b_i$$
\end{pro}
\begin{proof}{See proposition \ref{mu.mu.sum} (ii)}
\end{proof}

The following proposition contains a generalization of Proposition 2 of \cite{kmer-b}:\\

    \begin{pro} \label{matIdThm}
         Let $0\leq k\leq \ell$, $0\leq n\leq \ell$ and $v'\in V'_{\ell n}$. The following matrix identities hold.
           \begin{itemize}
            \item[\rm(i)] $A_{\ell k;B}^{\top}{\bf x}_{v'}^{\ell k n}=S_{\ell-k}(B(G_{v'}))\, {\bf z}_{v'}^{\ell n}$.
            \item[\rm(ii)] $A_{\ell k;B}{\bf z}_{v'}^{\ell n}={\bf x}_{v'}^{\ell k n}$.
            \item[\rm(iii)] $A_{\ell k;B}A_{\ell k;B}^{\top} {\bf x}_{v'}^{\ell k n}= S_{\ell-k}(B(G_{v'}))\, {\bf x}_{v'}^{\ell k n}$.
            \item[\rm(iv)] $A_{\ell k;B}^{\top}A_{\ell k;B} {\bf z}_{v'}^{\ell n}=S_{\ell-k}(B(G_{v'}))\, {\bf z}_{v'}^{\ell n}$.
            \item[\rm(v)] For any two distinct words $v\in V'_{\ell n_1}$ and $u\in V'_{\ell n_2}$, the vectors
            ${\bf x}_{v}^{\ell k n_1}$ and ${\bf x}_{u}^{\ell k n_2}$ are orthogonal.
            \item[\rm(vi)] For any two distinct words $v\in V'_{\ell n_1}$ and $u\in V'_{\ell n_2}$, the vectors
            ${\bf z}_{v}^{\ell n_1}$ and ${\bf z}_{u}^{\ell n_2}$ are orthogonal.
        \end{itemize}
 \end{pro}
\begin{proof}
    {The proofs of (i) and (ii) are concluded from definitions of ${\bf x}_{v'}$ and ${\bf z}_{v'}$ and Proposition \ref{someNuId}. Combining (i) and (ii) yields (iii) and (iv). The proofs of (v) is concluded from Proposition \ref{mu.mu.sum}(ii). The same proposition yields part (vi) by setting $k=\ell$.}
\end{proof}

\begin{thm}\label{egenvalueAAT,all} Let $0\leq k\leq \ell$. Then
\begin{itemize}
            \item[\rm(i)] The set $\{S_{\ell-k}(B(G_{v'})): v'\in  V'_{\ell, n; B},\, 0\leq n\leq \ell\}\}$ consists of all
eigenvalues of the matrix $A_{\ell, k;B}^{\top}A_{\ell, k;B}$ and the set $\{ {\bf z}_{v'}^{\ell n}: v' \in  V'_{\ell, n; B},\, 0\leq n\leq \ell\}$  is a complete set of eigenvectors corresponding to eigenvalues of $A_{\ell, k;B}^{\top}A_{\ell, k;B}$.
Moreover, these eigenvectors are pairwise orthogonal.
             \item[\rm(ii)] The set $\{S_{\ell-k}(B(G_{v'})): v'\in V'_{\ell, \leq k;B}\}$ consists of all non-zero
eigenvalues of the matrix $A_{\ell, k;B}A_{\ell, k;B}^{\top}$. The set
$\{{\bf x}_{v'}^{\ell k n}: v'\in V'_{\ell,\leq k;B}\}$ is a complete set of eigenvectors corresponding to aforementioned nonzero eigenvalues.
Moreover, these eigenvectors are pairwise orthogonal.
\end{itemize}
\end{thm}

\begin{proof}
{\begin{itemize}
            \item[\rm(i)] First we note that for any $v'\in  V'_{\ell, n;B}, 0\leq n\leq \ell$, $\, {\bf z}_{v'}^{\ell n}$ is a nonzero vector. By Proposition \ref{matIdThm} (iv), for any  $v'\in \Gamma_B$, ${\bf z}_{v'}^{\ell n}$ is an eigenvector of $A_{\ell, k;B}^{\top}A_{\ell, k;B}$  corresponding to the eigenvalue $S_{\ell-k}(B(G_{v'}))$. By Proposition \ref{matIdThm} (vi), these eigenvectors are pairwise orthogonal.
                Since $|\{ {\bf z}_{v'}^{\ell n}: v' \in  V'_{\ell, n; B},\, 0\leq n\leq \ell\}| = |\Gamma_B|$ and $|\Gamma_B| = |\Sigma_B|$ equals the size of the matrix
                $A_{\ell, k;B}^{\top}A_{\ell, k;B}$, we conclude that $\{{\bf  z}_{v'}^{\ell n}: v' \in  V'_{\ell, n; B},\, 0\leq n\leq \ell\}$  is a complete set of eigenvectors corresponding to eigenvalues of $A_{\ell,k;B}^{\top}A_{\ell,k;B}$, as required.

            \item[\rm(ii)]  The set of all non-zero eigenvalues of $A_{\ell, k;B}A_{\ell, k;B}^{\top}$ is the same as the set of all non-zero eigenvalues of $A_{\ell, k;B}^{\top} A_{\ell, k;B}$. In part (i), we obtained the complete set of eigenvalues of $A_{\ell,k;B}^{\top} A_{\ell,k;B}$. On the other hand, if $0\leq n\leq \ell$ and $v'\in V'_{\ell, n; B}$, then $S_{\ell-k}(B(G_{v'}))=0$ holds if and only if $k<n\leq \ell$.
                 Hence, the set $\{S_{\ell-k}(B(G_{v'})): v'\in V'_{\ell, \leq k;B}\}$ consists of all non-zero
                 eigenvalues of the matrix $A_{\ell, k;B}A_{\ell, k;B}^{\top}$ as well as the matrix $A_{\ell, k;B}^{\top} A_{\ell, k;B}$.
                 Moreover, using Proposition \ref{matIdThm} (iii), the
                 corresponding eigenvectors are $\{{\bf x}_{v'}^{\ell k n}: v'\in V'_{\ell,\leq k; B}\}$ and these vectors are pairwise orthogonal.
 \end{itemize}
}
\end{proof}

The following corollary is a straight conclusion of Theorem \ref{egenvalueAAT,all}.

\begin{cor}\label{nullAAt}
The set $\{ {\bf z}_{v'}^{\ell n}: v' \in  V'_{\ell, n; B},\, k< n\leq \ell\}$,
is a basis for the null-space of the matrix  $A_{\ell, k;B}^{\top}A_{\ell, k;B}$.
\end{cor}

Using Fact 3 in Section 5.6 of \cite{handbooklin} and Theorem \ref{egenvalueAAT,all}, the following corollary gives the reduced SVD of $A_{\ell, k;B}$.

\begin{cor}\label{SVD cor}
Let $r= {\rm rank} (A_{\ell, k;B}^{\top} A_{\ell, k;B})=|V'_{\ell,\leq k;B}|$, $\{u_1, u_2, \ldots, u_r\}= \{\frac{1}{S_{\ell-k}(B(G_{v'}))}{\bf x}_{v'}^{\ell k n}: v'\in V'_{\ell,\leq k;B}\}$ and
$\{v_1, v_2, \ldots, v_r\}=\{ {\bf z}_{v'}^{\ell n}: v' \in  V'_{\ell, \leq k; B}\}$. Then the reduced SVD of $A_{\ell, k;B}$ is given by $A_{\ell, k;B}=U\Sigma V$, where
the columns of $U$ are vectors $u_i$, ($1\leq i \leq r$) and the columns of $V$ are vectors $v_i$, ($1\leq i \leq r$) and $\Sigma$ is an $r\times r$ diagonal matrix, whose diagonal entries are the corresponding nonzero
 eigenvalues of the matrix $A_{\ell, k;B}A^{\top}_{\ell, k;B}$.
\end{cor}

    \begin{definition} \label{Upsilon} Given $B=(b_1,\ldots,b_{\ell})$, we define $\Upsilon_{\ell, k;B}$ as a matrix whose rows and columns are indexed by the elements of $V_{\ell, k;B}$ and $V'_{\ell ,\leq k;B}$ respectively and whose entries are given by $\Upsilon_{\ell, k;B}(w,v')=(-1)^{|w|-|w|_g}\nu_B(w,v')$. In other words, the columns of $\Upsilon_{\ell, k;B}$ are exactly the vectors ${\bf x}_{v'}$ for $v'\in V'_{\ell ,\leq k;B}$.
    Also we define the matrix $\Lambda$ by $\Lambda = {\rm diag}(S_{\ell-k}(B(G_{v'})))_{v'\in V'_{\ell,\leq k;B}}$.

\end{definition}

    \begin{remark} \label{eigenDecomp}
         Using Lemma {\rm \ref{W_up} (ii)} and Theorem {\rm \ref{egenvalueAAT,all}}, we obtain an orthonormal nonzero eigendecomposition for $A_{\ell, k;B} A_{\ell, k;B}^{\top}$.
    \end{remark}

\section{Bases for the null space and the row space of $A_{\ell,k;B}$}\label{rowspace}
In this section we give concrete bases for the null space and the row space of $A_{\ell,k;B}$.
The basis for the null space is obtained just as a corollary of the arguments of the Section \ref{eigen sec}.
The basis for the row space is obtained by a proper selection of some rows of $A_{\ell,k;B}$;
This gives a combinatorial interpretation for the previous formula about the rank of $A_{\ell,k;B}$;
The process of finding a basis for the row space of the matrix $A_{\ell, k;B}$
is similar to the one about incidence matrices presented in \cite{wilsoneuropean}.

 \begin{thm}\label{rowspacethm} Let $0\leq k\leq \ell$. Then
 \begin{itemize}
            \item[\rm(i)] The set $\{{\bf  z}_{v'}^{\ell n}: v' \in  V'_{\ell, n; B},\,
             k< n\leq \ell\}$, is a basis for the null-space of the matrix  $A_{\ell, k;B}$.
        \item[\rm(ii)] The matrix $A_{\ell,\leq k;B}$ has the same row space as $A_{\ell,k;B}$ and  $\rk(A_{\ell,k;B})=R_k(B)$.
        For $w\in V_{\ell,\leq k;B}$ denote by ${\bf r}_w$ the row of $A_{\ell,\leq k;B}$ indexed $w$. Then the set
$\{{\bf r}_w: w\in V'_{\ell,\leq k;B}\}$ is a basis for the row space of  $A_{\ell,k;B}$.
\end{itemize}
 \end{thm}

\begin{proof}
{\begin{itemize}
            \item[\rm(i)] First we observe that
            the vector ${\bf x}_{v'}^{\ell k n}$ is the zero vector if and only if $k< n \leq \ell$.
            Therefore, we have $ \{ {\bf z}_{v'}^{\ell n}: v' \in  V'_{\ell, n; B},\,
             k< n\leq \ell\} \subseteq {\rm null} (A_{\ell, k;B})$. Now, considering
             $ {\rm null} (A_{\ell, k;B}) \subseteq {\rm null} (A_{\ell, k;B}^{\top}A_{\ell, k;B})$ and using Corollary \ref{nullAAt}, we conclude that
              $ {\rm null} (A_{\ell, k;B}) = \{ {\bf z}_{v'}^{\ell n}: v' \in  V'_{\ell, n; B},
             \, k< n\leq \ell\}$, as required.
\item[\rm(ii)] Let $1\leq i \leq k$ and $v\in V_{\ell, i-1; B}$. Let $v_1=g$, say $v=g s$ with $s\in V_{\ell, i-1; B''},$ where $B''=(b_2,\ldots,b_{\ell})$. If $v=gs$ is matchable with some $u\in \Sigma_B$, then $u$ is matchable with exactly one of the words $0\, s,1\, s,\ldots,(b_1-1)\, s$; Otherwise $u$ is matchable with none of these words. Thus we have
\begin{equation}\label{rwrv}
{\bf r}_{gs}=\sum_{i=0}^{b_1-1} {\bf r}_{is}
\end{equation}
which gives the row indexed $v\in V_{\ell, i-1; B}$ as the summation of some rows indexed by elements $v'\in V_{\ell, i; B}$. Similar argument is true if we consider any gapped position of a word $v\in V_{\ell, i-1; B}$ (instead of its first position), thus we conclude ${\rm row}(A_{\ell, i-1;B}) \subseteq  {\rm row}(A_{\ell, i;B})$ for each $i$, $1\leq i \leq k$. Hence, $A_{\ell,\leq k;B}$ has the same row space as $A_{\ell,k;B}$.
In Theorem \ref{egenvalueAAT,all} (ii), we obtained a complete set of non-zero eigenvalues of $A_{\ell, k;B} A_{\ell, k;B}^{\top}$ which yields ${\rm rank} (A_{\ell, k;B} A_{\ell, k;B}^{\top})=|V'_{\ell,\leq k;B}|=R_k(B)$.
Since ${\rm rank} (A_{\ell, k;B} A_{\ell, k;B}^{\top})$ equals  ${\rm rank} (A_{\ell, k;B})$, we obtain ${\rm rank}(A_{\ell, k;B})=R_k(B)$. To provide a combinatorial interpretation of this, we first rewrite the equation (\ref{rwrv}) in the form ${\bf r}_{0s}={\bf r}_{gs}- \sum_{i=1}^{b_1-1} {\bf r}_{is}$ and then generalize this as below
\begin{equation}\label{r0mg}
{\bf r}_{0^m s}=\sum_{t\in \Gamma_{B'}} (-1)^{m-|t|_g}{\bf r}_{ts},
\end{equation}
where $0<m\leq \ell$, $B'=(b_{1},\ldots,b_m)$,$\,\, B''=(b_{m+1},\ldots,b_\ell)$ and $s\in \Delta_{B''}$. The proof of (\ref{r0mg}) is by induction on $m$ and left to the reader. When $s \in \Gamma_{B''}$, equation (\ref{r0mg}) gives the row indexed by the word $v=0^m s\in V_{\ell,\leq k;B}\setminus V'_{\ell,\leq k;B}$ as a linear combination of rows indexed by some words $v'\in V'_{\ell,\leq k;B}$. By a proper permutation of the $0$ positions of $v$, it is observed that the last statement is true for any word $v \in V_{\ell,\leq k;B}\setminus V'_{\ell,\leq k;B}$ regardless of its positions of zeros;  Hence, the set $\{{\bf r}_w: w\in V'_{\ell,\leq k;B}\}$ is a generator for the row space of $A_{\ell,k;B}$. Now, using ${\rm rank}(A_{\ell,k;B})=R_k(B)$, we conclude that
 $\{{\bf r}_w: w\in V'_{\ell,\leq k;B}\}$ is a basis for the row space of $A_{\ell,k;B}$.
\end{itemize}}
\end{proof}

\section{Computing the entries of the matrices $W_{\ell, k;B}$ and $H_{\ell, k;B}$}\label{w sec}
In this section we give a concrete description of the entries of matrices $W_{\ell, k;B}$ and $H_{\ell, k;B}$,
using the orthonormal nonzero eigendecomposition of the matrices $A_{\ell, k;B}A_{\ell, k;B}^{\top}$.

\begin{thm}\label{W-entries}
Let $u \in \Sigma_B$, $v \in V_{\ell,k;B}$. Moreover, with notation of Definition \ref{defPQ}, let $P=P(u,v)$ and $Q=Q(u,v)$. Then the entry $W_{\ell,k;B}(u,v)$, the Moore-Penrose pseudo-inverse of $A_{\ell, k;B}$, is given as below
\begin{equation}
\label{w-ent-f}
W_{\ell, k;B}(u,v)=\frac{1}{\displaystyle \prod_{i \in \overline{G}_v} b_i}\,\, \sum_{G,\, G_v \subseteq G \subseteq [\ell]}\frac{\displaystyle (-1)^{|Q\setminus G|}\prod_{i \in P\setminus G}(b_i-1)}{S_{\ell-k}(B(G))}
\end{equation}
\end{thm}

\begin{proof}
{For $v' \in V'_{\ell ,\leq k}$, let $d_{v'}=\frac{1}{||{\bf x}_{v'}||^2 \lambda_{v'}}$. By Lemma \ref{W_up} and the definitions of $A_{\ell,k;B}$ and $\Upsilon_{\ell,k;B}$ we obtain
\begin{align*}
	W_{\ell, k;B}(u,v) &=\sum_{y\in M_{\ell, k;B}(u)}\sum_{v'\in V'_{\ell ,\leq k}} \nu_B(v,v')\, \nu_B(y,v') d_{v'}\\
&=\sum_{v'\in V'_{\ell ,\leq k}}\nu_B(v,v')\, d_{v'} \sum_{y\in M_{\ell, k;B}(u)}\nu_B(y,v')\\
&=\sum_{v'\in V'_{\ell ,\leq k}}(-1)^{\ell-k}\nu_B(v,v')\, d_{v'}\, S_{\ell-k}(B(G_{v'}))\,\nu_B(u,v')~~~\hbox{(by  (\ref{nuId1}))}
\end{align*}
Replacing $d_{v'}$ using Proposition \ref{p-nrm-x}, we obtaine
\begin{align*}
W_{\ell, k;B}(u,v)&=(-1)^{\ell-k}\sum_{v'\in V'_{\ell ,\leq k}} \frac{\nu_B(v,v')\, \nu_B(u,v')}{S_{\ell-k}(B(G_{v'})) {\displaystyle \prod_{i\in \overline{G}_{v'}} (v'_i+{v'_i}^2)\prod_{i\in G_{v'}}b_i}}\\
&=(-1)^{\ell-k}\sum_{G, G_v \subseteq G}\,\,\, \sum_{v'\in V'_{\ell ,\leq k},G_{v'}=G} \frac{\nu_B(v,v')\, \nu_B(u,v')}{S_{\ell-k}(B(G_{v'})) {\displaystyle \prod_{i\in \overline{G}_{v'}} (v'_i+{v'_i}^2)\prod_{i\in G_{v'}}b_i}}\\
&=(-1)^{\ell-k}\sum_{G,\, G_v \subseteq G \subseteq [\ell]}\,\frac{1}{S_{\ell-k}(B(G)) {\displaystyle  \prod_{i\in G}b_i} }  \,\, \sum_{v'\in \Gamma_B\, ,G_{v'}=G} \frac{\nu_B(v,v')\,\nu_B(u,v')}{ {\displaystyle \prod_{i\in \overline{G}_{v'}} (v'_i+{v'_i}^2)}}\\
&=\frac{1}{\displaystyle \prod_{i \in \overline{G}_v} b_i}\,\, \sum_{G,\, G_v \subseteq G \subseteq [\ell]}\frac{\displaystyle (-1)^{|Q\setminus G|}\prod_{i \in P\setminus G}(b_i-1)}{S_{\ell-k}(B(G))} ~~~\hbox{(by (\ref{fracNuId}))}~~~
\end{align*}
 }
\end{proof}

\begin{thm}\label{G-entries} The sum of entries of any rows (columns) of the matrix $H_{\ell,k;B}:=W_{\ell,k;B}A_{\ell,k;B}$ equals $1$. Furthermore, for any $u,w \in \Sigma_B$, the entry $H_{\ell,k;B}(u,w)$ is given as below, where by using Definition \ref{defPQ}, $P=P(u,w)$ and $Q=Q(u,w)$.
\begin{equation}
\label{g-ent-f}
H_{\ell, k;B}(u,w)=\frac{1}{\displaystyle \prod_{i=1}^{\ell} b_i}\,\, \sum_{G \subseteq [\ell],\,\,\ell-k\leq |G|}{\displaystyle (-1)^{|Q\setminus G|}\prod_{i \in P\setminus G}(b_i-1)}
\end{equation}

\end{thm}
\begin{proof}
{
To compute the sum of entries of each row (column) of the matrix $H$, we observe that ${\bf z}_{g g \cdots g}^{\ell k 0}={\bf j}$. Now, using Proposition \ref{matIdThm} (iv) and Lemma \ref{W_nzed_A}, we have $H{\bf j}={\bf j}$, as desired.

To calculate the entry $H_{\ell,k;B}(u,w)$ of $H$, firstly note that if $v\sim w$ then for any subset $G_v\subseteq G \subseteq [\ell]$ we have
$P(u,v)\setminus G=P(u,w)\setminus G$ and $Q(u,v)\setminus G=Q(u,w)\setminus G$.
Secondly, by $H_{\ell,k;B}=W_{\ell,k;B}A_{\ell,k;B}$ we have
\begin{equation*}
H_{\ell,k;B}(u,w)=\sum_{v\in V_{\ell,k;B},\,\,v\sim w} W_{\ell,k;B}(u,v)
\end{equation*}

Thirdly, the result of Theorem \ref{W-entries}, is rewritten as
\begin{equation*}
W_{\ell,k:B}(u,v)=\frac{1}{\displaystyle \prod_{i=1}^{\ell} b_i}\,\, \sum_{G,\, G_v \subseteq G \subseteq [\ell]}\frac{\displaystyle (-1)^{|Q(u,v)\setminus G|}{\displaystyle \prod_{i \in G_v} b_i}\prod_{i \in P(u,v)\setminus G}(b_i-1)}{S_{\ell-k}(B(G))}
\end{equation*}
Thus by the two last formulas we obtain
\begin{align*}
H_{\ell,k;B}(u,w)&=\frac{1}{\displaystyle \prod_{i=1}^{\ell}b_i}\sum_{v\in V_{\ell,k;B},\,\,v\sim w}\,\,\,\, \sum_{G,\, G_v \subseteq G \subseteq [\ell]}\frac{\displaystyle (-1)^{|Q(u,v)\setminus G|}{\displaystyle \prod_{i \in G_v} b_i}\prod_{i \in P(u,v)\setminus G}(b_i-1)}{S_{\ell-k}(B(G))}
\\
&={\displaystyle \frac{1}{\prod_{i=1}^{\ell}b_i}}\,\sum_{G\subseteq [\ell],\,\ell-k\leq |G|}\,\,\,\, \sum_{v\in V_{\ell,k;B},\,v\sim w,G_v\subseteq G}\frac{\displaystyle (-1)^{|Q(u,v)\setminus G|}{\displaystyle \prod_{i \in G_v} b_i}\prod_{i \in P(u,v)\setminus G}(b_i-1)}{S_{\ell-k}(B(G))}\\
&={\displaystyle \frac{1}{\prod_{i=1}^{\ell}b_i}}\,\sum_{G\subseteq [\ell],\,\ell-k\leq |G|}\,\,\,\, \sum_{v\in V_{\ell,k;B},\,v\sim w,G_v\subseteq G}\frac{\displaystyle (-1)^{|Q(u,w)\setminus G|}{\displaystyle \prod_{i \in G_v} b_i}\prod_{i \in P(u,w)\setminus G}(b_i-1)}{S_{\ell-k}(B(G))}\\
&={\displaystyle \frac{1}{\prod_{i=1}^{\ell}b_i}}\,\sum_{G\subseteq [\ell],\,\ell-k\leq |G|}\frac{\displaystyle (-1)^{|Q(u,w)\setminus G|}\prod_{i \in P(u,w)\setminus G}(b_i-1)}{S_{\ell-k}(B(G))} \sum_{v\in V_{\ell,k;B},\,v\sim w,G_v\subseteq G}{\displaystyle \,\,\,\, \prod_{i \in G_v} b_i}\\
\end{align*}
Now, considering the fact that the inner summation is exactly $S_{\ell-k}(B(G))$, we obtain (\ref{g-ent-f}).
}
\end{proof}

{\bf Acknowledgement.} M. Mohammad-Noori would like to thank University of Tehran for supporting him during his sabbatical leave at IPM. The researches of N. Ghareghani and M. Mohammad-Noori were in part supported by grants from IPM (No. 94050016 and No. 95050129).

\end{document}